\documentclass[11pt]{article}

\usepackage{fullpage}
\usepackage{xspace}
\usepackage{paralist}
\usepackage{mathrsfs}
\usepackage{graphicx}
\usepackage{comment}
\usepackage{amsmath,amssymb,amsthm}

\usepackage{color}

\newtheorem{theorem}{Theorem}
\newtheorem{lemma}[theorem]{Lemma}
\newtheorem{corollary}[theorem]{Corollary}

\newtheorem{definition}[theorem]{Definition}

\newtheorem{observation}[theorem]{Observation}

\newcommand{\height}[1]{height(#1)}

\makeatletter
\def\exitNode{\@ifnextchar[{\@with}{\@without}}
\def\@with[#1]#2{\ensuremath{exitNode(#2,#1)}}
\def\@without#1{\ensuremath{exitNode(#1)}}
\makeatother

\newcommand{\localDH}{LOCAL-DOWNHILL\xspace}
\newcommand{\algname}{Forward-If-Empty\xspace}
\newcommand{\FTE}{FIE\xspace}
\newcommand{\localFTE}{LOCAL-\FTE}
\newcommand{\greedy}{GREEDY\xspace}

\title{Buffer Size for Routing\\ Limited-Rate Adversarial Traffic}
\author{Avery Miller \and Boaz Patt-Shamir}
\date{}

\begin{document}
	
\maketitle

	\vspace{-12mm}
	
	\begin{center}
		\small
		\begin{tabular}{cc}
%
			\multicolumn{2}{c}{School of Electrical Engineering}\\
			\multicolumn{2}{c}{Tel Aviv University}\\ 
			\multicolumn{2}{c}{Tel Aviv 6997801}\\ 
			\multicolumn{2}{c}{Israel}
		\end{tabular}
	\end{center}
	\vspace*{0.5cm}

	\begin{abstract}
		We consider the slight variation of the adversarial queuing theory
		model, in which an adversary injects packets with
		routes into the network subject to the following 
		constraint: For any link $e$, the total number of packets injected
		in any time window $[t,t')$ and whose 
		route contains  $e$, 
		is at most $\rho(t'-t)+\sigma$, where
		$\rho$ and $\sigma$ are non-negative parameters. Informally,
		$\rho$ bounds the long-term rate of injections and $\sigma$ bounds the
		``burstiness'' of injection: $\sigma=0$ means that the injection is
		as smooth as it can be.
		
		It is known that greedy scheduling of the packets (under which a link is not
		idle if there is any packet ready to be sent over it) may result in
		$\Omega(n)$ buffer size even on an $n$-line network and very smooth
		injections ($\sigma=0$). In this paper we propose a simple
		non-greedy scheduling policy and show that, in a tree where all
		packets are destined at the root, no buffer needs to be larger than
		$\sigma+2\rho$ to ensure that no overflows occur, which is  optimal in our 
		model.
		The rule of our algorithm is to \emph{forward a packet
			only if its next buffer is completely empty}. The policy is
		centralized: in a single step, a long 
		``train'' of packets may progress together. We show that in some
		sense central coordination is required, by presenting an injection
		pattern with $\sigma=0$ for the $n$-node line that results in
		$\Omega(n)$ packets in a buffer if local control is used, even for
		the more sophisticated ``downhill'' 
		algorithm, which forwards a packet only if its next buffer is less
		occupied than its current one.
	\end{abstract}
	
	\pagenumbering{arabic}

	\vspace*{0.5cm}
\section{Introduction}

We study the process of packet routing over networks, 
where an adversary injects
packets at nodes, and the routing algorithm is required to forward the
packets
along network links until they reach their prescribed destination,
subject to link capacity constraints. Our guiding question is the
following: assuming  that the injection pattern adheres to some given
upper bound specification, what is the smallest
buffer size that will allow a routing algorithm
to deliver all traffic, i.e., that will ensure
that there is no overflow at the node buffers?

To bound the injection rate, we follow the classical
$(\sigma,\rho)$ burstiness 
model of Cruz \cite{Cruz,network-calculus-book}. Specifically, we assume that the number of
packets that are injected in any time interval of $t$ time units is at most
$\sigma+\rho\cdot t$ for some non-negative parameters $\sigma$ and
$\rho$. Intuitively, $\rho$ represents the maximal long-term injection rate and
$\sigma$ represents the maximal ``burstiness'' of the injection pattern.
We assume that the injection pattern is feasible, in the sense
that the total average rate of traffic that needs to cross a link does not
exceed its capacity. 

Our approach is very close to that of the \emph{adversarial  queuing
	theory}~\cite{aqt,aqt2} henceforth abbreviated AQT. 
In the AQT model, packets are injected 
along with their routes by an adversary.  The adversary is limited by the
feasibility constraint, which is formalized as follows. Assuming that
all links have capacity $1$, i.e., a link can deliver at most one packet at
a time step, the
requirement is that 
in any time window of length $w$, the number of injected
packets that need to use any link does not exceed $\lceil
w\rho\rceil$, where  $w$ 
and $\rho$ are model parameters.%
\footnote{
	This model is almost equivalent to Cruz's $(\sigma,\rho)$ model
	(see discussion in  \cite{aqt}).
	We chose to use the $(\sigma,\rho)$ model as it allows for
	simpler expressions to bound the buffer size.
} The main question in AQT is
when a given routing policy is \emph{stable}, i.e., what is the maximum rate that allows the
queues to be bounded under the given policy. 
Furthermore, AQT concerns itself with
\emph{local}, \emph{greedy} policies. Local policies are defined by a
rule that can be applied by each node based only on it local
information (packets residing in that node and possibly its immediate
neighbors). Greedy
(a.k.a.\ 
\emph{work conserving}) policies are policies under which a link is
not fully utilized only when there are not enough packets ready to be transmitted over that
link. These restrictions are justified by the results that say that there
are local greedy policies that are stable for any feasible injection rate
\cite{aqt2}. While stability means that the required buffer size can
be bounded, 
the bounds are usually large (polynomial in the network size).
It should be noted that even linear buffer size is not practical in most cases. Furthermore, it is known 
	that
	in the Internet, big buffers have
	negative effect on traffic (cf. ``bufferbloat'' \cite{bufferbloat}).

In this work, in contrast to AQT, on one hand we are interested in the
quantitative question of buffer size, and on the other hand we do not
restrict ourselves 
to local greedy policies. 
While the interest in buffer size is obvious,  we offer the following
justification to our liberalism
regarding the nature of policies we consider.
First, we claim that  
with the advent of
software-defined networks~\cite{sdn}, central control over the routing algorithm
has become a reality in many networks and should not be disqualified
as a show-stopper anymore. Moreover, our results give a strong indication
that insisting on strict locality may result in a significant blowup in
buffer size. 
Our second relaxation, namely not insisting on greedy
policies, is not new, as 
it is already known that greedy policies
may require large buffers. Specifically,  in \cite{RS}, Ros\'{en} and
Scalosub show that the buffer size required to ensure no
losses in an $n$-node line network with a single
destination is $\Theta(n\rho^2)$, where $\rho$ is the injection rate.
Note that this result  means that for
greedy routing, 
sublinear buffers can guarantee loss-free routing only if the injection
rate is $o(1)$.

We present two {sets of} results. Our main result is positive: we 
propose a centralized
routing algorithm that requires buffers whose size is independent of the network size. 
We prove that in the case of tree networks, when packets are destined
to the tree root,  if the injection pattern is
feasible with injection rate $\rho$ and maximal burstiness $\sigma$, then the required buffer
space need not exceed $\sigma + 2\rho$ in order not to lose any packet. 
{We provide a matching lower bound to show that this is the optimal 
	buffer size.} The
routing algorithm, which we believe to be attractive from the
practical point of view, says simply ``forward a packet to the next hop only
if its next buffer is \emph{empty}.'' 
The algorithm is
centralized in that it may simultaneously forward 
long ``trains'' of packets (a train consists of a single packet per
node and an empty buffer in front of
the leading packet). 

Our second {set of results} is negative. We show that central 
coordination
is necessary, 
even for the more refined (and non-greedy) \emph{downhill algorithm}.
In the downhill algorithm, a packet is forwarded over a link whenever the buffer at the other 
endpoint is less full than the buffer in its current location.
We show that even in the $n$-node
line network,
there are feasible injection patterns with $0$
burstiness under which the downhill algorithm results in buffer
buildup of $\Omega(n)$ packets. Interestingly, as we show, there are 
certain situations where the downhill algorithm requires buffers of 
size $\Omega(n)$ while the greedy algorithm only needs buffers of size 
1, and other situations where the downhill algorithm needs buffers of 
size 1 while the greedy algorithm needs buffers of size $\Omega(n)$.

We may note in this context the result of Awerbuch
\emph{et al.}~\cite{ABBS}, where they consider
the case of a
single destination node in dynamic networks.  They show  that
	a certain variant of the downhill algorithm ensures that the number of
	packets in a buffer is bounded by $O(n S_{\max})$, where 
	$S_{\max}$ is
	a bound on the number of packets co-residing in the network under
	an unknown optimal schedule.
\subsection{More related work}

In the \emph{buffer management} model \cite{MPL,KLMPSS}, a different
angle is taken.
The idea is to lift all restrictions on
the 
injection pattern, implying that packet loss is possible.  The goal is to deliver as many 
packets as possible at their destination.  The buffer management model is usually used to
study routing 
within a single switch, modeled by very simple
topologies (e.g., a single buffer \cite{AMRR}, a star \cite{KM}, or a complete bipartite graph \cite{KR}). The
difficulty in such scenarios may be due to   
packets with different values or to some dependencies between packets
\cite{EHMPRR}.  The tree topology is studied in \cite{KLMP}.
The interested reader is referred to \cite{ovflow-survey} for a
comprehensive survey of the work in this area.

The idea of the downhill algorithm has been used for various objectives (avoiding 
deadlocks \cite{ms-1}, computing maximal flow \cite{gt},
multicommodity flow \cite{al1,al2}, and routing in the context of
dynamically
changing topology \cite{AAGMRS,APV,ABBS,KOR}). With the
exception of \cite{KOR}, the buffer size is usually assumed to be linear in the
number of nodes (or in the length of the longest possible simple path
in the network). In \cite{KOR}, a buffer's height is accounted for by
counters, so that each node needs to hold only a constant number of
packets and $O(\log n)$-bit counters.

\paragraph{Organization of the paper.} The remainder of the paper is structured as follows. In Section \ref{model}, we define the network model. Section \ref{ForwardToEmpty} describes and analyzes our \algname algorithm. Section \ref{sectlocaldownhill} gives lower bounds for local downhill algorithms.

\section{Model}\label{model}


\paragraph{The system.}
We model the system as a directed graph $G=(V,E)$, where
nodes represent hosts or network switches (routers), and edges
represent communication links. For a link $u\to v$ (denoted by 
$(u,v)$), we say that $u$ is
an \emph{in-neighbor} of $v$, and $v$ is an \emph{out-neighbor} of $u$.
Each edge $e\in E$ has capacity
$c(e)\in\mathbb{N}$. We consider static systems, i.e., $G$ is fixed
throughout the execution.

\paragraph{Input (a.k.a.\ adversary).}
We assume that in every {round} $t$, a set of packets is \emph{injected}
into the system. Each packet $p$ is injected at a node, along with a
complete route that specifies a simple path, denoted by $r(p)$, in $G$ that starts at the node of
injection, called the \emph{source}, and ends at the packet's
\emph{destination}. The set of all packets injected in the time interval $[t,t']$ is denoted by $P(t,t')$. We use $P(t)$ to denote the set of packets injected in {round} $t$. Many packets may be injected in the same node, and
the routes may be arbitrary. However, we consider the following type of
restriction on injection patterns. 
\begin{definition}
	For any $\sigma,\rho\ge 0$,
	an injection pattern is said to adhere to a $(\rho,\sigma)$ bound if,
	in any time interval $[t,t']$ and for any edge $e\in E$, it holds
	that
	$
	\left|\left\{p\in P(t,t')\mid e\in r(p)\right\}\right|\le \rho(t'-t)+\sigma
	$.
\end{definition}

\paragraph{Executions.}
A system \emph{configuration} is a mapping of packets to nodes. An
execution of the system is an infinite sequence of configurations
$C_0,C_1,C_2\ldots$, where $C_i$ for $i>0$ is called the configuration after
round $i$, and $C_0$ is the initial configuration. The evolution of
	the system from $C_i$ to $C_{i+1}$ consists of a sequence of
	\emph{ministeps}: 0 or more \emph{injection ministeps} followed by 0 or more \emph{forwarding ministeps}. In particular, in any round $i$,
	\begin{compactenum}
		\item There is one injection ministep for each packet $p\in P(i)$, and, in this ministep, $p$ is mapped to the first node in
		$r(p)$.
		\item In each forwarding ministep, each packet $p$ currently mapped to a node $v$ such that $e=(v,u)$ is an
		edge in $r(p)$, may be re-mapped to $u$. {Such a re-mapping occurs at the end of round $i$, and, during subsequent ministeps before the end of the round, $p$ is considered ``in transit" and is not mapped to any node.} In this case we say that
		$v$ \emph{forwards} $p$ over $e$ to $u$ in round $i$. If $u$
		is the destination of $p$, then $p$ is said to be \emph{delivered}
		and is removed from subsequent configurations. The number of packets
		forwarded over $e$ in one round may not exceed $c(e)$. The choice of which
		packets to forward is controlled by the \emph{algorithm}.
	\end{compactenum}
	
	We use $v[s]$ to refer to node $v$ at the start of ministep $s$. We view packets mapped to a node as stored by the node's \emph{buffer}, which is
	an array of \emph{slots}. The {\emph{load}} of a buffer in a given
	configuration is the number of packets mapped to it in that
	configuration. We further assume that in a configuration, each packet
	mapped to a node is mapped to a particular slot in that node's buffer.
	The slot's index
	is called the packet's \emph{position}, denoted $pos(p)$
	for a packet $p$. Note that we do not place explicit restrictions on the buffer
	sizes, so overflows never occur. 
	
	\paragraph{Algorithms.}
	The role of the algorithm is to determine which packets are forwarded
	in each round. The algorithm must obey the link capacity
	constraints. We distinguish between two types of algorithms: local and
	centralized. In a \emph{centralized} algorithm, the decisions of the
	algorithm may depend, in each round, on 
	the history of complete system configurations at that time. In a \emph{local}
	algorithm, the decision of which 
	packets should be forwarded from node $v$ may depend only on the
	packets stored in the buffer of $v$ and the packets stored in the
	buffers of $v$'s neighbors. Note that both centralized and local
	algorithms are required to be on-line, i.e., may not make decisions
	based on future injections.

	\paragraph{Target Problem: Information Gathering on Trees.}
	In this paper, we consider networks whose underlying topology is a
	directed tree where all links have the same capacity $c$. Furthermore, we assume that the
	destination of all packets is the root of the tree. We sometimes call
	the root the \emph{sink} of the system. Injections adhere to a $(\rho,\sigma)$ bound. To ensure that the 
		injection pattern is feasible, i.e., that finite buffers suffice to 
		avoid
		overflows, we assume that $\rho \leq c$.

\section{The \algname Algorithm}\label{ForwardToEmpty}

In this section, we describe the \algname (\FTE) algorithm and prove that, for any $(\rho,\sigma)$-injection pattern, the load of every buffer is bounded above by $\sigma + 2\rho$ during the execution of \FTE. We then show that this bound is optimal by proving that, for any algorithm, there is an injection pattern such that the load of some buffer reaches $\sigma+2\rho$. 


\subsection{Algorithm}

First, we specify how \FTE positions packets within each buffer. Each buffer is partitioned into \emph{levels} of $c$ slots each, where level $i \geq 1$ consists of slots $(i-1)c + 1,\ldots,ic$. For a packet $p$, its level is given by $\lceil pos(p) / c \rceil$. {Suppose that $m \geq 1$ packets are mapped to a node, then the algorithm maps the packets to positions $1,\ldots,m$.}

\begin{definition}
	The \emph{height} of a node $v$ at ministep $s$ is denoted by \height{v[s]} and is defined as $\max_p\{\textrm{level of packet $p$ in $v[s]$}\}$. The height of the sink is defined to be $-\infty$.
\end{definition}

Next, we specify how the algorithm behaves during each forwarding ministep. Intuitively, we think of the system as having sections of ``ground" that consist of connected subgraphs of nodes with height at most 1, and ``hills" that consist of connected subgraphs of nodes with height greater than 1. The algorithm works by draining the ground packets ``horizontally" towards the sink, and by breaking off some packets from the boundaries of the hills to fall into empty nodes surrounding the hills. Notice that nothing happens in the interior of each hill, e.g., the peak does not get flattened; in each forwarding ministep, only packets from the boundary of the hill get chipped away.

We now describe the algorithm in detail. Each round consists of $c$ forwarding ministeps. At the start of every forwarding ministep $s$, the algorithm computes a maximal set $AP(s)$ of directed paths called \emph{activation paths} for ministep $s$. (The set $AP(s)$ might be empty.)  All nodes that are contained in paths of $AP(s)$ are considered \emph{activated} for ministep $s$. In what follows, we say that two paths are \emph{node-disjoint} if their intersection is empty or equal to the sink. To construct $AP(s)$, the algorithm greedily adds maximal directed paths to $AP(s)$ and ensures that all paths in $AP(s)$ are node-disjoint. The paths can be one of three types:
\begin{compactenum}
	\item \emph{Downhill-to-Sink:} the first node has height greater than 1, the last node is the sink, and all other nodes in the path (if any) have height 1.
	\item \emph{Downhill-to-Empty:} the first node has height greater than 1, the last node has height 0, and all other nodes in the path (if any) have height 1.
	\item \emph{Flat:} the last node is the sink or has height 0, and all other nodes in the path have height 1.
\end{compactenum}
At the start of ministep $s$, the algorithm first adds Downhill-to-Sink paths to $AP(s)$ until none remains, then adds Downhill-to-Empty paths to $AP(s)$ until none remains, then adds all Flat paths to $AP(s)$. Each time a path is added to $AP(s)$, the nodes in that path (except the sink) are unusable in the remainder of the construction of $AP(s)$. For any two different forwarding ministeps $s$ and $s'$, a node (or even a path) can be used in both $AP(s)$ and $AP(s')$, even if $s$ and $s'$ are in the same round.

In each forwarding ministep, each activated node for that ministep
forwards {1 packet (if it has any).}
Since the packets are identical, it does not matter which packet is forwarded. However, it is convenient for our analysis to assume that the buffers are LIFO, i.e., the forwarded packets 
are taken from an activated node's highest level and, when received, 
are stored in the receiving node's highest level.  This ensures that 
the level of a packet can only change when it is forwarded.

\subsection{Analysis}

We now prove that, for any $(\rho,\sigma)$-injection pattern, the load of every buffer is bounded above by $\sigma + 2\rho$ during the execution of \FTE. Without loss of generality, we may assume that $\rho = c$, since, if $\rho < c$, we can artificially restrict the edge capacity to $\rho$ when executing the algorithm to get the same result.

%
%

The analysis of our algorithm depends crucially on what happens to the 
loads of connected subsets of nodes that all meet a certain minimum 
height requirement. Informally, we can think of any configuration of 
the system as a collection of hills and valleys, and, for a given hill, 
``slice" it at some level $h$ and look at what happens to all of the 
packets in that hill above the slice. We call the portion of the hill that has 
packets above this slice a \emph{plateau}.

\begin{definition}
	A \emph{plateau of height $h$} (or \emph{$h$-plateau}) at ministep $s$ is a maximal set of nodes forming a connected subgraph such that each node has a packet at level $h-1$ and at least one of the nodes has a packet at level $h$. A plateau $P$ at ministep $s$ will be denoted by $P[s]$, and $\height{P[s]}$ denotes its height at the start of ministep $s$.
\end{definition}

\begin{observation}\label{intersection}
	If $P[s],Q[s]$ are plateaus at ministep $s$, then they are disjoint or one is a subset of the other. If $\height{P[s]} = \height{Q[s]}$ and $P[s] \neq Q[s]$, then $P[s] \cap Q[s] = \emptyset$.
\end{observation}

Plateaus are defined so that every 
pair of disjoint plateaus is separated by a sufficiently deep
``valley." This ensures that a packet forwarded from one plateau does 
	not immediately arrive 
	at another plateau, allowing us to argue about the number of packets in each plateau 
independently. We  measure the ``fullness" of a plateau by the number of 
packets it contains at a given level and higher.


\begin{definition}\label{kload}	
	The \emph{$k$-load} of a plateau $P[s]$ is defined to be the number of packets in $P[s]$ at level $k$ or higher. For any plateau $P[s]$, we denote by $load_k(P[s],s')$ the $k$-load of $P[s]$ at time $s'$.
\end{definition}

In Definition \ref{kload}, note that $P[s]$ might not be a plateau at the start of a ministep $s'$, but it still represents a well-defined set of nodes. Note that we may speak of the $k$-load of an $h$-plateau for any $k$ and $h$.

We now proceed to analyze the dynamics of the \FTE algorithm. The first thing to note about the choice of activated nodes in each forwarding ministep is that, whenever a node receives packets at the end of a round, the packets are always placed at level 1 in the node's buffer. In particular, each node receives at most $c$ packets, and, either the receiving node's buffer was empty at the beginning of the round, or, it only had packets at level 1 and forwarded as many packets as it received. This observation leads to the following result, which we state without proof for further reference.

\begin{lemma}\label{arrivesatone}
	Let $w$ be a non-sink node. Packets forwarded to $w$ in round $i$ are at level 1 in node $w$'s buffer at the end of round $i$.
\end{lemma}

Lemma \ref{arrivesatone} highlights a crucial property of our algorithm. The remainder of our analysis is dedicated to bounding the total number of packets at level 2 or higher, and Lemma \ref{arrivesatone} guarantees that any packet that is forwarded will immediately drop to level 1. Further, the algorithm ensures that the level of a packet will never increase. Informally, this means that once the algorithm chooses to forward some packet $p$, we know that $p$ will never again contribute to the 2-load of the system.

We now set out to prove that the load of each 
2-plateau decreases by $c = \rho$ packets per round due to 
	forwarding, which means 
that any increase in the load of a plateau is due to the burstiness of injections.

 Since the network is a directed tree, we know that, from every node, there exists a unique directed path to the sink. So, we can uniquely identify how packets will leave each plateau, which we now formalize.

\begin{definition}
	Consider any $h$-plateau $P[s]$ at the start of a forwarding ministep 
	$s$. The node $v$ in $P[s]$ whose outgoing edge leads to 
	a node whose height is at most $h-2$ at the start of ministep $s$ is 
	called the \emph{exit node} of $P[s]$ at ministep $s$ and is denoted by 
	$\exitNode{P[s],s}$. We define the 
	\emph{landing node} of $P[s]$ at ministep $s$ to be the node at the 
	head of the outgoing edge of $\exitNode{P[s],s}$.
\end{definition}

In Figure \ref{plateaus}, we illustrate the definitions and concepts introduced so far by giving an example of how the system evolves in a forwarding step {when $c=1$.}

\begin{figure}[!ht]
	\begin{center}
		\includegraphics[scale=0.4]{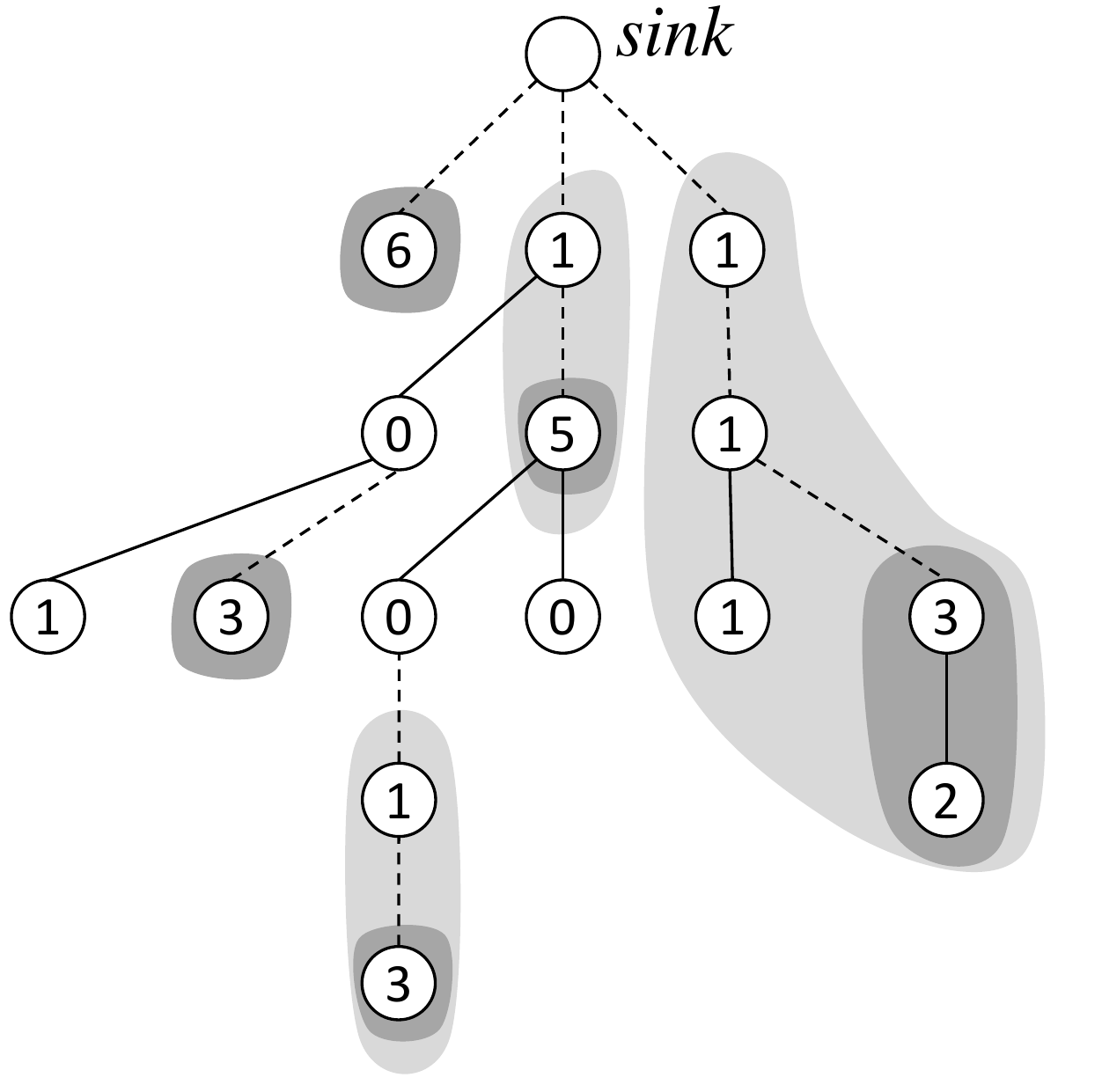}
		\hspace{2cm}
		\includegraphics[scale=0.4]{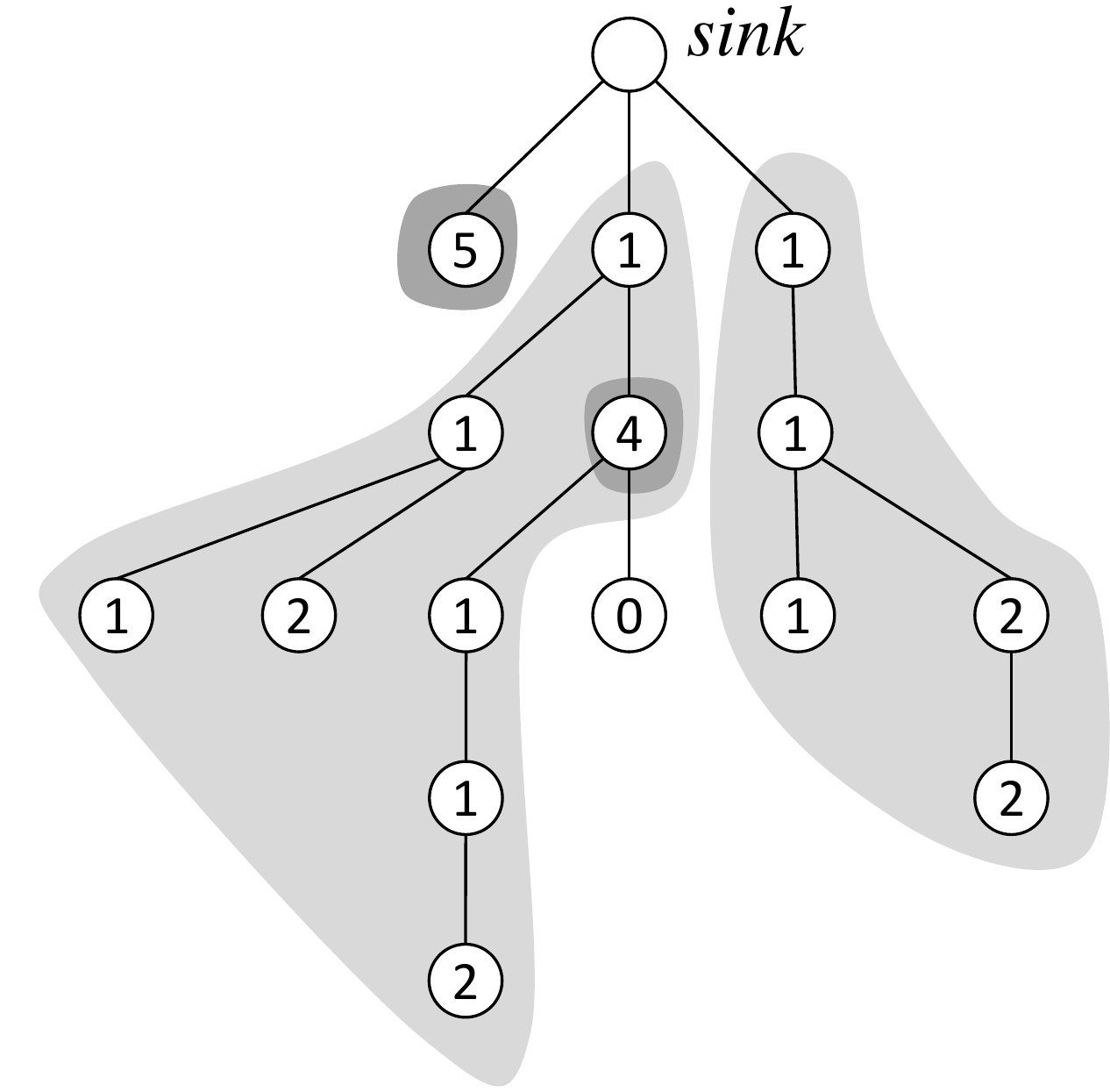}
	\end{center}
	\caption{On the left, the state of a network at the start of some forwarding ministep. Each node is labeled with its load. Dotted lines show the activation paths chosen by the algorithm. On the right, the network immediately after this ministep. Each light grey region is a 2-plateau, and each dark grey region is a subplateau with height greater than 2. In each plateau, the node closest to the sink is the exit node and its parent is the landing node.}
	\label{plateaus}
\end{figure}

We want to show that, for each forwarding ministep $s$, the 2-load of each 2-plateau shrinks by at least $1$. The following result shows that this is the case if the exit node of such a plateau is activated during $s$. 

\begin{lemma}\label{exitforwards}
	For any forwarding ministep $s$ and any $2$-plateau $P[s]$, we have $load_{2}(P[s],s+1) \leq load_{2}(P[s],s)$. If $\exitNode{P[s],s}$ is activated in $s$, then $load_{2}(P[s],s+1) \leq load_{2}(P[s],s) - 1$.
\end{lemma}
\begin{proof}
	First, for any forwarding ministep $s$ that is not the last one in a round, it is easy to see that $load_{2}(P[s],s+1) \leq load_{2}(P[s],s)$ since no forwarded packets arrive at any nodes until the end of the round. If $s$ is the last forwarding ministep $s$ of a round, then Lemma \ref{arrivesatone} guarantees that any packet arriving at a node is stored in level 1, which implies that the number of packets stored at level 2 or higher does not increase.
	
	Next, if $\exitNode{P[s],s}$ is activated in ministep $s$, we claim that there exists 1 packet that is at level at least $2$ in $P[s]$ at the start of $s$ that is forwarded in ministep $s$. If the height of $\exitNode{P[s],s}$ is at least 2, then the packet forwarded from $\exitNode{P[s],s}$ satisfies our claim. Otherwise, let $w$ be the landing node of $P[s]$. By the definition of $\exitNode{P[s],s}$, $w$ is either the sink or a node of height 0. From the specification of the algorithm, there is a path of activated nodes starting from a node $v$ in $P[s]$ with height at least 2, ending at $w$, and all nodes other than $v$ and $w$ have height exactly 1. The packet forwarded from $v$ satisfies our claim. Finally, we note that if $s$ is the final forwarding ministep of a round, then Lemma \ref{arrivesatone} guarantees that any packet arriving at a node is stored in level 1, which implies that the number of packets stored at level 2 or higher does not increase.
\end{proof}

The next challenge is to deal with the fact that plateaus can merge during a round. For example, an exit node of one plateau might forward packets to a node $v$ that was sandwiched between two disjoint plateaus both with larger height than $v$, but the increase in $v$'s height results in one large plateau formed by the merging of $v$ with the two plateaus. Figure \ref{plateaus} shows an example of three plateaus merging into one. We want to compare the load of this newly-formed plateau with the loads of the plateaus that merged. In particular, for any round $i$, we consider its forwarding ministeps $s_1,\ldots,s_1+c-1$, and for any plateau $P[s_1+c]$, we consider its ``pre-image" to be the set of plateaus that existed at the start of ministep $s_1$ that merged together during ministeps $s_1,\ldots,s_1+c-1$ to form it.

\begin{definition}
	For any round $i$, consider the forwarding ministeps $s_1,\ldots,s_1+c-1$. Consider any $h$-plateau $P[s_1+c]$ with $h \geq 2$ that exists at the start of ministep $s_1+c$. The \emph{pre-image of $P[s_1+c]$}, denoted by \emph{Pre($P[s_1+c]$)}, is the set of $h$-plateaus that existed at the start of $s_1$ that are contained in $P[s_1+c]$.
\end{definition}

The following observation follows from the fact that any two distinct plateaus with the same height are disjoint (see Observation \ref{intersection}).

\begin{observation}\label{disjointpre}
	For any plateau $P[s_1+c]$ that exists at the start of ministep $s_1+c$, any two distinct plateaus $P_1[s_1],P_2[s_1] \in Pre(P[s_1+c])$ are disjoint. Further, for any two distinct $h$-plateaus $P_1[s_1+c],P_2[s_1+c]$ that exist at the start of ministep $s_1+c$, we have $Pre(P_1[s_1+c]) \cap Pre(P_2[s_1+c]) = \emptyset$.
\end{observation}

We now prove a technical lemma that guarantees that, for any 2-plateau $P[s_1+c]$, in each forwarding ministep $s \in \{s_1,\ldots,s_1+c-1\}$ there always exists a plateau in $Pre(P[s_1+c])$ that contains an exit node that is activated during ministep $s$. This will imply that, in each ministep during the merge, the 2-load of one of the 2-plateaus involved in the merge decreases by at least 1. 

\begin{lemma}\label{oneactivated}
	For any round $i$, consider the forwarding ministeps $s_1,\ldots,s_1+c-1$. For any 2-plateau $P[s_1+c]$ that exists at the start of ministep $s_1+c$, for each $s \in \{s_1,\ldots,s_1+c-1\}$, there exists a $P[s] \subseteq P[s_1] \in Pre(P[s_1+c])$ such that $\exitNode{P[s],s}$ is activated in ministep $s$.
\end{lemma}
\begin{proof}
	First, note that since $P[s_1+c]$ is a 2-plateau, there is a non-empty set $S$ of nodes in $P[s_1+c]$ such that each node in $S$ has height at least 2 at the start of ministep $s_1+c$. Lemma \ref{arrivesatone} implies that each node in $S$ has height at least 2 at the start of each ministep $s_1,\ldots,s_1+c-1$. For each $s \in \{s_1,\ldots,s_1+c-1\}$, let $\mathscr{P}[s]$ be the set of 2-plateaus that contain at least one node from $S$ at the start of ministep $s$.

Consider an arbitrary $s \in \{s_1,\ldots,s_1+c-1\}$. From $\mathscr{P}[s]$, choose a plateau $P[s]$ that minimizes the distance between $\exitNode{P[s],s}$ and the sink. Let $P[s_1]$ be the plateau in $Pre(P[s_1+c])$ that contains $P[s]$. If $\exitNode{P[s],s}$ is activated during ministep $s$, we are done. Otherwise, consider the landing node $\ell$ of $P[s]$. Since $P[s]$ is a 2-plateau, then, by definition, $\ell$ is the sink or $\ell$'s buffer is empty at the start of ministep $s$. The set of activation paths chosen by the algorithm is maximal, so there is some path $Q$ of nodes with $\ell$ as its last node. Further, since Flat paths have lowest priority, we know that $Q$ is either a Downhill-To-Sink path (if $\ell$ is the sink) or a Downhill-To-Empty path (if $\ell$ isn't the sink). Therefore, the first node in $Q$ has height at least 2, the last node in $Q$ is $\ell$, and all other nodes (if any) have height exactly 1. Let $P'[s]$ be the 2-plateau that contains $Q \setminus \{\ell\}$ at the start of ministep $s$, and note that $\exitNode{P'[s],s}$ is activated in ministep $s$ since all nodes in $Q$ are activated. Hence, it is sufficient to show that $P'[s]$ is contained in some plateau of $Pre(P[s_1+c])$. To do so, note that, at the start of ministep $s_1+c$ (i.e., after all packets that were forwarded in round $i$ have been received) all nodes in $Q \cup \{\ell\}$ have height at least 1. This is because the first node of $Q$ had greater than $c$ packets at the start of ministep $s$ and forwards at most one packet per ministep, and each other node in $Q \cup \ell$ had one packet forwarded to it in ministep $s$. Further, all nodes in $P[s] \subseteq P[s_1]$ have height at least 1 at the at the start of ministep $s_1+c$ since, by the definition of $Pre(P[s_1+c])$, we know that $P[s_1] \subseteq P[s_1+c]$ is a 2-plateau. Since $\ell$ is the last node of $Q$ and the landing node of $P[s]$, it follows that the nodes of $Q \cup P[s] \cup \{\ell\}$ are all contained in the same 2-plateau at the start of ministep $s_1+c$. Since $P[s] \subseteq P[s_1] \subseteq P[s_1+c]$, it follows that $Q \subseteq P[s_1+c]$, and so $Q$ is contained in some plateau $P'[s_1]$ of $Pre(P[s_1+c])$. It follows that $P'[s] \subseteq P'[s_1] \in Pre(P[s_1+c])$, as desired.
\end{proof}

The next result shows that, when plateaus merge during a forwarding ministep, at least $c$ packets descend to level 1 in the merging process. 

\begin{lemma}\label{forwardlose}
	For any round $i$, consider the forwarding ministeps $s_1,\ldots,s_1+c-1$. For any 2-plateau $P[s_1+c]$ that exists at the start of ministep $s_1+c$, we have $load_{2}(P[s_1+c],s_1+c) \leq load_{2}( \bigcup_{P[s_1] \in Pre(P[s_1+c])} P[s_1],s_1) - c$. 
\end{lemma}
\begin{proof}
	To calculate $load_2(P[s_1+c],s_1+c)$, we can sum up the 2-loads of all nodes in $P[s_1+c]$ at the start of ministep $s_1+c$. By Lemma \ref{arrivesatone}, any packet with level at least 2 in some node $v$ at the start of ministep $s_1+c$ was also at level at least 2 in $v$ at the start of each ministep $s_1,\ldots,s_1+c-1$, which means that $v$ belongs to some $P[s_1] \in Pre(P[s_1+c])$. In particular, each node of $P[s_1+c]$ with height at least 2 is in some $P[s_1] \in Pre(P[s_1+c])$, and, each node in every $P[s_1] \in Pre(P[s_1+c])$ is contained in $P[s_1+c]$ (by definition). It follows that the value of $load_2(P[s_1+c],s_1+c)$ is equal to $\sum_{P[s_1] \in Pre(P[s_1+c])} load_2(P[s_1],s_1+c)$. 
	
	For each $s \in \{s_1,\ldots,s_1+c-1\}$, and for each $P[s_1] \in Pre(P[s+1])$, we can apply Lemma \ref{exitforwards} to conclude that either:
	\begin{compactenum}
		\item $load_2(P[s_1],s+1) \leq load_2(P[s_1],s)$ (if $P[s_1]$ does not contain a 2-plateau whose exit node is activated in ministep $s$), or,
		\item $load_2(P[s_1],s+1) \leq load_2(P[s_1],s) - 1$ (if $P[s_1]$ does contain a 2-plateau whose exit node is activated in ministep $s$).
	\end{compactenum}
	By Lemma \ref{oneactivated}, there exists some $P'[s_1] \in Pre(P[s+1])$ that contains a plateau $P[s]$ such that $\exitNode{P[s],s}$ is activated in ministep $s$. This implies that there is always some $P'[s_1]$ whose 2-load decreases by at least 1, for each ministep $s \in \{s_1,\ldots,s_1+c-1\}$. So, by an induction argument over the possible values of $s$, we can show that 
	\[
	\sum_{P[s_1] \in Pre(P[s_1+c])} load_2(P[s_1],s_i) \leq \sum_{P[s_1] \in Pre(P[s_1+c])} load_2(P[s_1],s_1) - (i-1)
	\]
	Setting $i = c+1$ gives that 
	\[
	load_2(P[s_1+c],s_1+c) = \sum_{P[s_1] \in Pre(P[s_1+c])} load_2(P[s_1],s_1+c) \leq load_2(\bigcup_{P[s_1] \in Pre(P[s_1+c])} P[s_1],s_1) - c
	\]
\end{proof}

We now arrive at the main result, which shows that the buffer loads are always bounded above by $\sigma+2c$.

\begin{theorem}
	Consider the execution of \algname in a directed tree network with link capacity $c$. Suppose that the destination of all injected packets is the root, and that the injection pattern adheres to a $(c,\sigma)$ bound. For each node $v$, the number of packets in $v$'s buffer never exceeds $\sigma+2c$.
\end{theorem}
\begin{proof}
	For any ministep $s$, let $\mathit{round}(s)$ be the round that 
	contains $s$. 
	Let $\mathit{flat}(s)$ be the last ministep, up to and including $s$, 
	such that 
	all nodes have height at most 1 at the start of the ministep. For any 
	ministeps $s_1 \leq s_2$, let $\mathit{totalinj}(s_1,s_2)$ be the total 
	number of packets injected into the system in all ministeps in the 
	range $s_1,\ldots,s_2$. The following invariant bounds the total number of packets at level 2 or higher.
	\begin{quotation}
		\hspace{-6mm}\textbf{Invariant I}: at the start of any ministep $s$, the sum of 
		2-loads of all 2-plateaus is bounded above by 
		$\mathit{totalinj}(\mathit{flat}(s),s-1) - c \cdot (\mathit{round}(s) - 
		\mathit{round}(\mathit{flat}(s)))$.
	\end{quotation}
	To see why Invariant I is sufficient to prove the theorem, note that, at 
	the start of any ministep $s$, the height of any node $v$ is bounded 
	above by the 2-load of the 2-plateau containing $v$, plus the number of 
	packets at level 1 in $v$'s buffer. The 2-load of the 2-plateau 
	containing $v$ is bounded above by the sum of the 2-loads of all 
	2-plateaus, and the number of packets at level 1 in $v$'s buffer is 
	bounded above by $c$. By the invariant, the sum of the 2-loads of all 
	2-plateaus is bounded above by $\mathit{totalinj}(\mathit{flat}(s),s-1) 
	- c \cdot 
	(\mathit{round}(s)-\mathit{round}(\mathit{flat}(s))$. By separately 
	considering 
	the packets injected in the first 
	$(\mathit{round}(s)-\mathit{round}(\mathit{flat}(s)))$ rounds and the 
	packets 
	injected so far in $\mathit{round}(s)$, we get 
	\[\mathit{totalinj}(\mathit{flat}(s),s-1) 
	\leq c \cdot (\mathit{round}(s)-\mathit{round}(\mathit{flat}(s))) + 
	\sigma + 
	c\] Putting together all of these facts, we get that the height of any 
	node $v$ is bounded above by \[c \cdot 
	(\mathit{round}(s)-\mathit{round}(\mathit{flat}(s))) + \sigma + c - c 
	\cdot 
	(\mathit{round}(s)-\mathit{round}(\mathit{flat}(s))) + c \leq \sigma + 
	2c\]
	
	We now proceed to prove the invariant. Clearly, the invariant holds at 
	the start of any ministep where every node has height at most 1. Assume 
	that the invariant holds for all ministeps up to and including ministep 
	$s$, and assume that, at the start of ministep $s+1$, at least one node 
	has a packet at level 2 at the start of the ministep. This last assumption implies that $\mathit{flat}(s) = \mathit{flat}(s+1)$. 
	
	There are three cases to consider.  If $s$ is a forwarding ministep, there are two cases. 
	\begin{enumerate}
		\item Suppose that $s$ is an 
		injection ministep. Then exactly one packet is injected into the system 
		during $s$, so the sum of the 2-loads of all 2-plateaus increases by at 
		most 1. On the other hand, the value of 
		$\mathit{totalinj}(\mathit{flat}(s),s-1) - c 
		\cdot (\mathit{round}(s) - \mathit{round}(\mathit{flat}(s)))$ increases 
		by 
		exactly 1 since only the first term can change during an injection 
		ministep. This shows that the invariant holds at the start of ministep $s+1$.
		
		\item Suppose that $s$ is not the final forwarding ministep of the round. Then the sum of the 2-loads of all 2-plateaus does not increase since no packet is injected and any forwarded packets are not received until the end of the round. On the other hand, since $\mathit{round}(s) = 
		\mathit{round}(s+1)$, the value of $\mathit{totalinj}(\mathit{flat}(s),s-1) - c 
		\cdot (\mathit{round}(s) - \mathit{round}(\mathit{flat}(s)))$ does not decrease. Therefore, the invariant still holds at the start of ministep $s+1$.
		
		\item Suppose that $s$ is the final forwarding ministep of the round. That is, if the forwarding ministeps of $\mathit{round}(s)$ are numbered $s_1,\ldots,s_1+c-1$, then $s = s_1+c-1$. Let $\mathscr{P}[s_1+c]$ be the set of 2-plateaus that exist at the start of ministep $s_1+c$. 
		
		For each $P[s_1+c] \in \mathscr{P}[s_1+c]$, we apply Lemma \ref{forwardlose} to conclude that 
		\begin{equation}\label{eachloses}
			load_{2}(P[s_1+c],s_1+c) \leq load_{2}( \bigcup_{P[s_1] \in Pre(P[s_1+c])} P[s_1],s_1) - c.
		\end{equation}
		Applying inequality \ref{eachloses} over all $P[s_1+c] \in \mathscr{P}[s_1+c]$, we get 
		\begin{equation}\label{combinelosses}
			load_{2}( \bigcup_{P[s_1+c] \in \mathscr{P}[s_1+c]} P[s_1+c],s_1+c) \leq load_{2}( \bigcup_{P[s_1+c] \in \mathscr{P}[s_1+c]}\bigcup_{P[s_1] \in Pre(P[s_1+c])} P[s_1],s_1) - c.
		\end{equation}
		Next, by Observation \ref{disjointpre}, we know that all of the pre-images of plateaus in $\mathscr{P}[s_1+c]$ are pairwise disjoint. In particular, this means we can simply bound the sum of the loads of these pre-images by the sum of the loads of all plateaus in $\mathscr{P}[s_1]$. Namely,
		\begin{equation}\label{unionbound}
			load_2(\bigcup_{P[s_1+c] \in \mathscr{P}[s_1+c]}\bigcup_{P[s_1] \in Pre(P[s_1+c])} P[s_1],s_1) \leq \sum_{P[s_1] \in \mathscr{P}[s_1]} load_2(P[s_1],s_1).
		\end{equation}
		Therefore, by combining inequalities \ref{combinelosses} and \ref{unionbound}, we have shown that 
		\begin{equation}\label{maininequality}
			load_2( \bigcup_{P[s_1+c] \in \mathscr{P}[s_1+c]} P[s_1+c],s_1+c) \leq \sum_{P[s_1] \in \mathscr{P}[s_1]} load_2(P[s_1],s_1) - c.
		\end{equation}
		Since the invariant holds at the start of ministep $s$,
		\begin{equation}\label{boundPS}
			\sum_{P[s_1] \in \mathscr{P}[s_1]} load_2(P[s_1],s_1) \leq 
			\mathit{totalinj}(\mathit{flat}(s_1),s_1-1) - c \cdot (\mathit{round}(s_1) - 
			\mathit{round}(\mathit{flat}(s_1))).
		\end{equation}
		
		Next, since we assumed that there is a node with height at least 2 at the start of ministep $s+1 = s_1+c$, it follows from Lemma \ref{arrivesatone} that there is a node with height at least 2 in each of the ministeps $s_1,\ldots,s_1+c-1$ as well. It follows that $\mathit{flat}(s_1) = \mathit{flat}(s_1+c)$. Further, since no packets are injected in ministeps $s_1,\ldots,s_1+c-1$ and $\mathit{round}(s_1) = 
		\mathit{round}(s_1+c) - 
		1$, it follows that
		\begin{equation*}
			\begin{split}
				& \mathit{totalinj}(\mathit{flat}(s_1),s_1-1) - c \cdot (\mathit{round}(s_1) 
				- 
				\mathit{round}(\mathit{flat}(s_1))) \\
				& = \mathit{totalinj}(\mathit{flat}(s_1+c),s_1+c-1) - c \cdot 
				(\mathit{round}(s_1+c) - 1 - 
				\mathit{round}(\mathit{flat}(s_1+c)))\\
				& = \mathit{totalinj}(\mathit{flat}(s_1+c),s_1+c-1) - c \cdot 
				(\mathit{round}(s_1+c) - 
				\mathit{round}(\mathit{flat}(s_1+c)))+c.
			\end{split}
		\end{equation*}

		Combining this fact with inequalities \ref{maininequality} and \ref{boundPS}, we conclude that
		\[
		load_2( \bigcup_{P[s_1+c] \in \mathscr{P}[s_1+c]} P[s_1+c],s_1+c) \leq 
		\mathit{totalinj}(\mathit{flat}(s_1+c),s_1+c-1) - c \cdot 
		(\mathit{round}(s_1+c) - 
		\mathit{round}(\mathit{flat}(s_1+c))).
		\]
		Hence, we have shown that the invariant holds at the start of ministep $s+1$.
		
	\end{enumerate}

\end{proof}

\subsection{Existential optimality}
In this section, we show that there is no algorithm that can prevent buffer overflows when the buffer size is strictly less than $\sigma+2c$. 
\begin{theorem}
Consider the network consisting of nodes $v_1,\ldots,v_n$ such that, for each $i \in \{1,\ldots,n-1\}$, there is a directed edge from $v_i$ to $v_{i+1}$. Consider any edge capacity $c \geq 1$ and $\sigma \geq 0$. For any algorithm, there is a $(c,\sigma)$-injection pattern such that the load of some buffer is at least $\sigma + 2c$ in some configuration of the algorithm's execution. 
\end{theorem}
\begin{proof}
 Consider the following injection pattern that injects packets destined for node $v_n$. In the first round, inject $c$ packets into $v_1$, and, in each subsequent round $i$, injects $c$ packets into the node with smallest index that has the maximum load at the start of round $i$. If the maximum load in round $i-1$ was $m_{i-1}$ at some node $v_j$, then either $v_j$ or $v_{j+1}$ has load at least $\lceil m_{i-1}/2 \rceil$ after the forwarding ministep of round $i-1$. Therefore, at the start of each round $i > 1$, the load at some node is at least $\lceil \sum_{k=1}^{i-1} c/2^{k} \rceil$. So, at the start of some round $i$, some node $v$ has load at least $c$, at which time $c+\sigma$ packet injections are performed at $v$.
\end{proof}

\section{Local Downhill Algorithms}\label{sectlocaldownhill}

In this section, we consider downhill algorithms where each node in the network must decide whether or not it will forward packets based only on local information. We show  that local downhill 
algorithms require significantly larger buffer size than the 
centralized algorithm presented in Section~\ref{ForwardToEmpty}. We 
also show that local downhill algorithms may require significantly 
larger buffer sizes than the \greedy algorithm where each node 
always forwards packets if it has any. However, we show that the 
opposite is also true by providing situations where \greedy 
requires significantly larger buffer sizes than local downhill 
algorithms.

\subsection{Local vs. Centralized Algorithms}\label{localVScentralized}
Recall that in a downhill algorithm, packets may be forwarded 
only to a lighter-load node. 
Intuitively, the difficulty faced by local downhill algorithms is that 
they cannot perform a coordinated move in which all nodes along a path 
simultaneously forward a packet each. 
In this case, only the node closest to the sink knows whether or not it 
will forward a packet (based on the load of its out-neighbour) so, to 
ensure that no packets are forwarded ``uphill", none of the other nodes 
in the interval can decide to forward a packet. We now set out to show the effects of this limitation on a local version of \FTE and on a more sophisticated downhill algorithm.

\paragraph{The bad scenario.}
We consider a network $v_1,\ldots,v_n$ of nodes (where $n>2$) such 
that, for each $i \in \{1,\ldots,n-1\}$, there is a directed edge with 
capacity 1 from $v_i$ to $v_{i+1}$. One packet with destination $v_n$ 
is injected at node $v_1$ in each round. Note that, under this 
$(1,0)$-injection pattern, \greedy only needs buffers of size 1 to 
ensure that no overflows occur.

\paragraph{\localFTE.} The \localFTE algorithm is a local version of  
\algname  
described in Section \ref{ForwardToEmpty}, defined as follows. 
Each node 
forwards a packet if and only if its buffer is non-empty and its 
out-neighbour has an empty buffer. Under the injection pattern specified above, $v_1$ receives a packet in every 
round but only forwards a packet every two rounds. It follows that, in 
the first $R$ rounds, the load of $v_1$ is $\lceil R/2 \rceil$. Note that $R$ 
can be chosen to be arbitrarily large, independently of $n$. 
This can be generalized to any injection pattern with $\rho>1/2$.

\begin{theorem}
	If $n > 2$, then, for any $R$ and constant $\rho>1/2$, there is a 
	$(\rho,1)$-injection pattern of $R$ packets such that \localFTE 
	requires buffers of size $\Omega(R)$ to prevent overflows.
\end{theorem}

\paragraph{\localDH.} In \localDH, each node forwards a packet if and 
only if 
its buffer is non-empty and the load of its out-neighbour's buffer is 
strictly less than the load of its own buffer. We set out to show that, 
after sufficiently many rounds, $v_1$'s load is $\Omega(n)$. 

For each $j \geq 1$, let \emph{state $S_j$} be the sequence of $n$ integers corresponding to the load of each node's buffer immediately after all injections of round $j$ occur.  In each state, we focus on what happens between $v_1$ and the first empty node. More formally, for any sequence $S$ of non-negative integers, the \emph{initial segment} of $S$, denoted by $init(S)$, is the maximal prefix of non-zero entries in $S$.
In what follows, we assume that $n$ is large enough such that $init(S_j)$ is always a proper subsequence of $v_1,\ldots,v_n$ (we later show that, for an execution of $R$ rounds, $n \geq \sqrt{R}+2$ suffices.) For each $k \geq 0$, we use $f(k)$ to denote the index of the first state such that $v_1$'s load is at least $k$. For each $k \geq 1$, we define $\Delta_k = f(k)-f(k-1)$. An intuitive definition of $\Delta_k$ is the number of forwarding ministeps that occur between $S_{f(k-1)}$ and $S_{f(k)}$. 

To prove the lower bound, we consider the states where the load of $v_1$ increases and show that the number of rounds between consecutive such states keeps growing by 2. We first illustrate this phenomenon by examining a concrete example.  In Figure \ref{localdownhill}, we have provided a prefix of the execution of \localDH. Let's consider the number of rounds between $S_7$ (the state where $v_1$'s load first becomes 3) and $S_{13}$ (the state where $v_1$'s load first becomes 4). Notice that, if we ignore $v_1$ in states $S_7,\ldots,S_{12}$, the remainder of their initial segments is equal to the initial segments of $S_2,\ldots,S_7$, respectively. Within this interval, notice that $S_3,\ldots,S_7$ is the part of the execution where $v_1$'s load first becomes 2 and then first becomes 3. So we can split up the set of states $S_7,\ldots,S_{13}$ as follows: 4 ``middle'' rounds corresponding to the states $S_3,\ldots,S_7$, plus the first round and last round. The above example illustrates the following general fact, which we will later formally prove.

\begin{lemma}\label{intervals}
	For every $k \geq 1$, $\Delta_{k+1} = \Delta_k + 2$.
\end{lemma}

\begin{figure}[!ht]
	\begin{center}
		\includegraphics[scale=0.56]{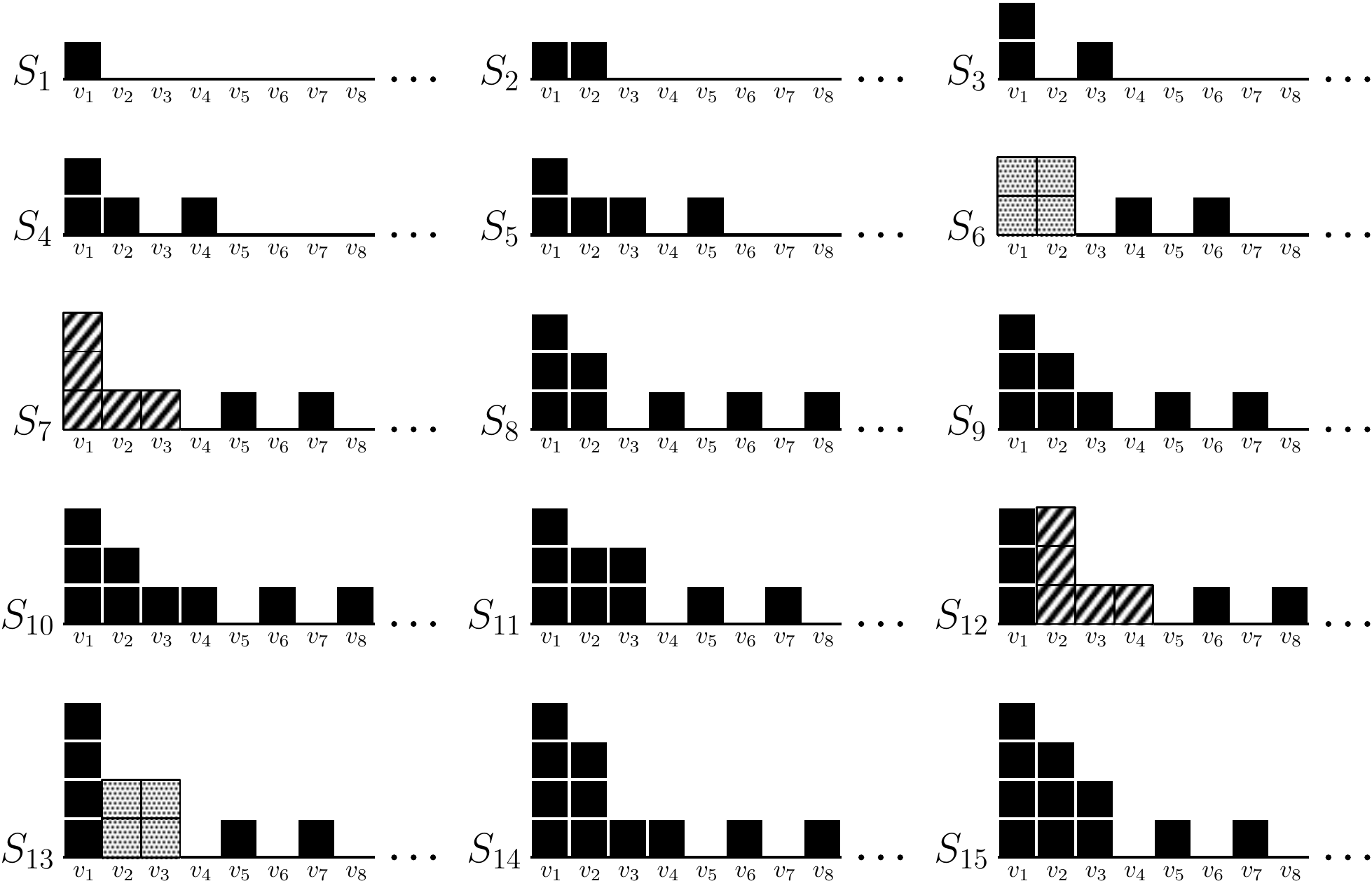}
	\end{center}
	\caption{The first 15 states of \localDH when a packet is injected at $v_1$ in each round. The dotted packets demonstrate that $tail(S_{f(4)}) = init(S_{f(3)-1})$. The diagonally-hatched packets demonstrate that $tail(S_{f(4)-1}) = init(S_{f(3)})$. }
	\label{localdownhill}
\end{figure}

The key to proving Lemma \ref{intervals} is observing, as we did in the above example, that the initial segments of some states will reappear later as the suffix of others. To describe this phenomenon, we divide the initial segment of each state $S$ into two parts: the first node of the segment is denoted by 
$\mathit{front}(S)$, and $\mathit{tail}(S)$ is defined to be
$\mathit{init}(S)$ 
without its first entry\footnote{LISP programmers might prefer to use \emph{car} and \emph{cdr} instead of \emph{front} and \emph{tail}!}.

We now provide the formal details of the lower bound. First, note that since a node only forwards a packet if the next node in the network has a strictly smaller load, the level of a packet never increases. Along with the fact that injections only happen at node $v_1$, we observe the following.

\begin{observation}\label{nonincreasing}
	For any state $S_j$, any subsequence of $init(S_j)$ is non-increasing.
\end{observation}

Next, we observe that the load of a node $v$ cannot increase from one state to the next unless $v$'s out-neighbour has the same initial load.

\begin{observation}\label{increasek}
	For any non-negative integer $k$, let $S$ be any non-increasing sequence with $front(S) = k$, and let $S'$ be the sequence that corresponds to executing \localDH with one packet arriving at the first node of $S$. Then $front(S') = k+1$ if and only if the second entry of $S$ is $k$.
\end{observation}

We now provide a technical lemma that will be used to prove the main result. It will be used to describe what happens to the tail of some initial segment in the case where the front node of the state does not forward a packet.

\begin{lemma}\label{firstandprev}
	For any positive integer $k$, suppose that $tail(S_{f(k+1)-1}) = init(S_{f(k)})$ and $tail(S_{f(k+1)}) = init(S_{f(k)-1})$. If one round of \localDH is applied to $S_{f(k+1)}$ with no packets arriving at its first node, then the initial segment of the resulting state $S$ is $init(S_{f(k+1)-1})$.
\end{lemma}
\begin{proof}
	First, we note that when \localDH is applied to $S_{f(k+1)}$, a packet is forwarded from the first node to the second. To see why, we must convince ourselves that the second entry of $S_{f(k+1)}$ is strictly less than $k+1$. Recall that $S_{f(k+1)}$ is obtained by applying \localDH to $S_{f(k+1)-1}$ with a packet injected at the first node. But, in $S_{f(k+1)-1}$, all entries are strictly less than $k+1$, and no packet is directly injected to the second node, so there is no way for a packet to be at level $k+1$ at the second node of $S_{f(k+1)}$.
	
	Since the front node of $S_{f(k+1)}$ fowards a packet, and no packets arrive at the first node of $S_{f(k+1)}$ by assumption, it follows that the front of $S$ is equal to $k$. Further, since $tail(S_{f(k+1)}) = init(S_{f(k)-1})$ by assumption, the tail of $S$ is exactly what we obtain by applying \localDH to $S_{f(k)-1}$ with a packet arriving at the first node. By definition, this means that the tail of $S$ is equal to $init(S_{f(k)})$. By assumption, $tail(S_{f(k+1)-1}) = init(S_{f(k)})$, so $tail(S_{f(k+1)-1}) = tail(S)$. Together with the fact that $front(S) = k$, it follows that $S = init(S_{f(k+1)-1})$.
\end{proof}

We can now prove the main fact needed for the proof of Lemma \ref{intervals}. The patterned packets in Figure \ref{localdownhill} demonstrate Lemma \ref{structure} in the case where $k=3$. 

\begin{lemma}\label{structure}
	For any positive integer $k$, $tail(S_{f(k+1)}) = init(S_{f(k)-1})$ and $tail(S_{f(k+1)-1}) = init(S_{f(k)})$.
\end{lemma}
\begin{proof}
	The proof is by induction on $k$. For the base case, consider $k=1$. Note that the first state whose front equals 2 is $(2,0,1,0\ldots,0)$, and the tail of this state is the empty sequence (since the initial segment consists of a single entry equal to 2). Also, the state immediately preceding the first state whose front equals 1 is $(0,\ldots,0)$, and the initial segment of this state is the empty sequence. Therefore, $tail(S_{f(2)}) = init(S_{f(1)-1})$. Next, note that the state immediately preceding $(2,0,1,0\ldots,0)$ is $(1,1,0,\ldots,0)$ and that its tail is $(1)$. Also, the first state whose front equals 1 is $(1,0,\ldots,0)$ and its initial segment is $(1)$. Therefore, $tail(S_{f(2)-1}) = init(S_{f(1)})$.
	
	As induction hypothesis, suppose that, for some $k \geq 1$, we have $tail(S_{f(k+1)}) = init(S_{f(k)-1})$ and $tail(S_{f(k+1)-1}) = init(S_{f(k)})$. We first prove that $tail(S_{f(k+2)-1}) = init(S_{f(k+1)})$ and then that $tail(S_{f(k+2)}) = init(S_{f(k+1)-1})$. 
	
	In the states $S_{f(k+1)},\ldots,S_{f(k+2)-1}$, note that the front of each state is equal to $k+1$ due to the arrival of a packet at node $v_1$ in each round. By Observation \ref{increasek}, $S_{f(k+2)-1}$ is the first of these states whose second entry is equal to $k+1$. 
	
	First, to prove that $tail(S_{f(k+2)-1}) = init(S_{f(k+1)})$, consider the application of \localDH starting at state $S_{f(k+1)}$ and stopping when $S_{f(k+2)-1}$ is reached. In each of these rounds, a packet is forwarded from the first node to the second node since the second entry in these states is strictly less than $k+1$. So, in each of these states, the tail is a non-increasing sequence to which we are applying \localDH with a packet arriving at the first node. By the induction hypothesis, the first such tail, i.e. $tail(X_{k+1}) = tail(S_{f(k+1)})$, is equal to $init(S_{f(k)-1})$. Starting from $S_{f(k)-1}$ and applying \localDH with a packet arriving at the first node, the first sequence we reach with front equal to $k+1$, by definition, is $S_{f(k+1)}$. Since $S_{f(k+2)-1}$ is the first state we reach whose second entry is equal to $k+1$, it must be the case that $init(S_{f(k+1)}) = tail(S_{f(k+2)-1})$.
	
	Next, to prove that $tail(S_{f(k+2)}) = init(S_{f(k+1)-1})$, consider what happens when \localDH is applied to $S_{f(k+2)-1}$ with a packet injected at the first node. The resulting state is $S_{f(k+2)}$, and we wish to determine the structure of its tail. Since the first two entries of $S_{f(k+2)-1}$ are $k+1$, no packet is forwarded from the first node to the second node. Therefore, the resulting tail is equal to what is obtained by applying \localDH to $tail(S_{f(k+2)-1})$ with no packets arriving at its first node. However, we just showed that $tail(S_{f(k+2)-1}) = init(S_{f(k+1)})$, so it follows that $tail(S_{f(k+2)})$ is equal to what is obtained by applying \localDH to $init(S_{f(k+1)})$ with no packets arriving at its first node. Since, by the induction hypothesis, we have $tail(S_{f(k+1)}) = init(S_{f(k)-1})$ and $tail(S_{f(k+1)-1}) = init(S_{f(k)})$, Lemma \ref{firstandprev} implies that $tail(S_{f(k+2)}) = init(S_{f(k+1)-1})$.
\end{proof}

We now prove Lemma \ref{intervals}, which states that, for $k \geq 1$, $\Delta_{k+1} = \Delta_k + 2$. Recall that $\Delta_k = f(k)-f(k-1)$.

\begin{proof}
	(Lemma \ref{intervals}) First, for the case where $k=1$, notice that $f(0)=f(1)=1$ and $f(2) = 3$. Therefore, $\Delta_1 = 0$ and $\Delta_2 = 2$, as required.
	
	For any $k \geq 2$, consider the application of \localDH starting at state $S_{f(k)}$ and stopping when $S_{f(k+1)-1}$ is reached. In each of these rounds, a packet is forwarded from the first node to the second node since the second entry in these states is strictly less than $k$. So, in each of these states, the tail is a non-increasing sequence to which we are applying \localDH with a packet arriving at the first node. By Lemma \ref{structure}, the first such tail is $init(S_{f(k-1)-1})$,  so the second such tail is $init(S_{f(k-1)})$, and we stop when $tail(S_{f(k+1)-1})$ is reached. By Lemma \ref{structure}, this last tail is equal to $init(S_{f(k)})$. Subtracting the first index from the last, we get $f(k)-f(k-1)+1 = \Delta_k+1$. Finally, it takes one more round to reach $S_{f(k+1)}$ from $S_{f(k+1)-1}$, which gives the desired result.
\end{proof}

We now show that $v_1$'s load is eventually $\Omega(n)$. Using Lemma \ref{intervals}, we can show that the index of the first state where $v_1$'s load is equal to $k$ is equal to the sum of the first $k$ positive even integers, which implies the following result.

\begin{corollary}\label{ReachHeight}
	For any positive integer $k$, we have $f(k) = k^2 - k + 1$.
\end{corollary}

By choosing $k = n-2$ and executing \localDH on the specified injection pattern for $R = k^2 - k + 1$ rounds, Corollary \ref{ReachHeight} implies that the buffer at node $v_1$ will contain $n-2$ packets in the last round of the execution. A straightforward induction argument shows that the width of the initial segment is bounded above by one more than the load of $v_1$, so our assumption that the initial segment is always a proper subsequence of $v_1,\ldots,v_n$ holds. Therefore, we get the following lower bound on the buffer size required by \localDH.

\begin{theorem}
	There is a (1,0)-injection pattern such that \localDH requires buffers of size $\Omega(n)$ to prevent overflows.
\end{theorem}

\subsection{Downhill vs.\ Greedy}\label{downhillVSgreedy}

In this section, we show that the comparison between local downhill 
algorithms and the \greedy algorithm is not one-sided. For the same 
$(1,0)$-injection pattern used to prove the lower bounds for \localFTE 
and \localDH in Section \ref{localVScentralized}, \greedy only needs 
buffers of size 1 to prevent overflows. However, we now show that there 
is also a $(1,0)$-injection pattern for which buffers of size 1 suffice 
for \localFTE and \localDH but \greedy requires buffers of linear size.

\begin{theorem}
	There is a (1,0)-injection pattern such that \localFTE and \localDH require buffers of size 1 while \greedy requires buffers of size $\Omega(n)$ to prevent overflows.
\end{theorem}
\begin{proof}
	Assume that $n$ is even. In the first phase of the injection pattern, we inject a single packet at node $v_{2i-1}$ in each of the rounds $i = 1,\ldots,n/2$. In the second phase of the injection pattern, a single packet is injected at node $v_{n-1}$ in rounds $n/2+1,\ldots,n$.
	
	In the execution of \greedy, at the end of the first phase, nodes $v_{n/2},\ldots,v_{n-1}$ each have exactly 1 packet in their buffer. In each round of the second phase, two packets arrive at node $v_{n-1}$: one packet forwarded from $v_{n-2}$ and one injected packet. However, only one packet is forwarded by $v_{n-1}$ in each round, which means that, at the end of round $n$, the buffer at $v_{n-1}$ contains $n/2$ packets.
	
	In the execution of \localFTE, at the end of the first phase, there is a packet at node $v_{n-1}$ and all other packets injected so far are at nodes with smaller indices (and no node ever has more than 1 packet in its buffer). In each round of the second phase, a single packet is injected at node $v_{n-1}$ and one packet is forwarded from $v_{n-1}$ to $v_n$ (which is immediately absorbed since $v_n$ is the sink). No packets are forwarded from $v_{n-2}$ to $v_{n-1}$ in this phase since the buffer at $v_{n-1}$ always contains a packet. It follows that at most one packet is in each node's buffer during the entire execution. The analysis of \localDH is identical.
\end{proof}

\section{Conclusion}
\label{sec-conc}

In this work, we have shown that an extremely simple
algorithm (requiring central coordination) suffices to bound the
required buffer size by a function of the burstiness of the injected
traffic. In other words, it is possible to design an algorithm that does not increase the inherent
burstiness of the traffic. This required the algorithm to be
non-greedy and non-local, and this is not coincidental. 

In future work, we would
like to extend the results to topologies more general than trees,
and to multiple destinations.
It would also be interesting to determine whether or not there is a more general local algorithm, e.g., that sometimes sends packets `uphill', whose buffer size requirements depend only on the values of $\rho$ and $\sigma$. The semi-local variant is also intriguing:
suppose that nodes can coordinate within a certain range. How does this affect packet accumulation?

\section*{Acknowledgment}
We wish to thank the authors of \cite{dobrev} for pointing out a bug in a preliminary version of this work and for suggesting an approach for addressing it.

\bibliographystyle{splncs03}
\bibliography{downhill-disc}

\end{document}